\documentclass[ conference,a4paper]{IEEEtran}
\IEEEoverridecommandlockouts                              

\IEEEoverridecommandlockouts                              
\usepackage{graphicx, subfigure}
\usepackage{enumerate}
\usepackage{amsmath}
\usepackage{amssymb}
\usepackage{amsthm}
\usepackage[utf8]{inputenc}
\usepackage{subfigure}
\usepackage{sgame}
\usepackage{color}
\usepackage{stmaryrd}

\newtheorem{proposition}{Proposition}

\newtheorem{remark}{Remark}

\newtheorem{defi}{Definition}

\def\be{\begin{equation}}
\def\ee{\end{equation}}

\def\bearn{\begin{eqnarray*}}
\def\eearn{\end{eqnarray*}}
\def\bear{\begin{eqnarray}}
\def\eear{\end{eqnarray}}
\def\barr{\begin{array}}
\def\earr{\end{array}}
\usepackage{subfigure}%
\usepackage{epstopdf}

 \usepackage{mathtools}
\newcommand{\defeq}{\vcentcolon=}

\def\Prob{I\kern-0.30em P}

\begin{document}

\title{ Altruism in groups:  an evolutionary games approach$^{\diamond}$ \thanks{$^\diamond$ This work has been partially supported by the European Commission within the framework of the CONGAS project FP7-ICT-2001-8-317672}  }

\author{
Ilaria Brunetti$^{\star \dagger}\thanks{$^\star$CERI/LIA, University of Avignon, 339, chemin des Meinajaries,Avignon, France}$\thanks{$^\dagger$INRIA, B.P 93, 06902 Sophia Antipollis Cedex, FRANCE}, Rachid El-Azouzi$^\star$  and Eitan Altman$^\dagger$ 
}
\date{}

\maketitle

\begin{abstract}
We revisit in this paper the relation between evolution of species and the mathematical tool of evolutionary games, which has been used to model and predict it. We indicate known shortcoming of this model that restricts the capacity of evolutionary games to model groups of individuals that share a common gene or a common fitness function. In this paper we provide a new concept to remedy this shortcoming in the standard evolutionary games in order to cover this kind of behavior. Further, we explore the relationship between this new concept and Nash equilibrium or ESS. We indicate through the study of some example in the biology as Hawk-Dove game, Stag-Hunt Game and Prisoner’s Dilemma, that when taking into account a utility that is common to a group of individuals, the equilibrium structure may change dramatically. We also study the multiple access control in slotted Aloha based wireless networks. We analyze the impact of the altruism behavior on the performance at the equilibrium.
\end{abstract}

\section{Introduction}
Evolutionary games become a central tool for predicting and even design evolution in many fields. Its origins come from biology where it was  introduced by \cite{MaynSm-Price} to model conflicts among animals.  It differs from classical game theory by its focusing on the evolution dynamics of the fraction of members of the population that use a given strategy, and in the notion of Evolutionary Stable Strategy (ESS, \cite{MaynSm-Price}) which includes robustness against a deviation of a whole (possibly small) fraction of the population who may wish to deviate. This is in  contrast with the standard Nash equilibrium concept that only incorporates robustness against deviation of a single user.  It became perhaps the most important mathematical tool  for describing and modeling evolution since Darwin. Indeed, on the importance of the ESS for understanding the evolution of species, Dawkins writes in his book "The Selfish Gene" \cite{gene}: "we may come to look back on the invention of the ESS concept as one of the most important advances in evolutionary
theory since Darwin." He further specifies: "Maynard Smith's concept of the ESS will enable us, for the first time, to
see clearly how a collection of independent selfish entities can come to resemble a single organised whole.  In this paper, we identify inherent restrictions on the modeling capacity of classical evolutionary games apply. Recently, the  evolutionary game theory has become of increased interest to social scientists \cite{FH03}.  In computer science, evolutionary game theory is appearing, some  examples of applications  can be found  in  multiple access protocols \cite{tembineWiopt}, multihoming\cite{SAK} and resources  competition in the Internet \cite{Zheng}.

The starting point of this theory is a situation of a very large number of  local pairwise interactions between pairs of individuals that are 
randomly matched. In classical evolutionary games (EG), each individual
represents a selfish player in a non-cooperative
game that it plays with its randomly
matched adversary, and in which each player seeks to maximize its utility.
The originality of EG is in postulating that 
this utility is the Darwinian fitness.
The fitness should be understood as the relative rate at which the behavior
used by the individual will increase.
For a given share of behaviors in the population,
a behavior of an individual with a higher fitness would thus result
in a higher rate of its reproduction.  We make the observation that 
since classical EG associates with an individual both the interactions
with other individuals as well as the fitness, then it is restricted to
describing populations in which the individual is the one that is 
responsible for the reproduction and where the choice of
its own strategies is completely selfish.  In  biology,  in some species like bees
or ants, the one who interacts is not the one who reproduces. This implies that the Darwinian
fitness is related to the entire swarm and not to a single bee and thus, standard EG models excludes these species  in which the single individual
which reproduces is not the one that interacts with other 
individuals.  Furthermore, in many species, we find altruistic behaviors, which may hurt the individual adopting it, favouring instead the group he belongs to. Altruistic behaviors are typical of parents toward their children: they may incubate
them, feed them or protect them from predator’s at a high cost for themselves. Another example
can be found in flock of birds: when a bird sees a predator it gives an alarm call to warn the rest
of the flock, attracting the predators attention to itself. Also the stinging behavior of bees is an
altruistic one: it serves to protect the hive, but its lethal for the bee which strives. In human behavior,
many phenomena where individuals do care about other’s benefits in their groups or about
their intentions can be observed in the real word.  It must be admitted that some phenomena require an explanation in terms of genes which pursue their own interest to the disadvantage of the individual. Hence the assumption of selfishness becomes
inconsistent with the real behavior of individual in a population.

Founders of classical EG seem to have been well aware of this problem.
Indeed, Vincent writes in \cite{vincent}
"Ants seem to completely subordinate any
individual objectives for the good of the group. On the other hand, the social
foraging of hyenas demonstrates individual agendas within a tight-knit social
group (Hofer and East, 2003). As evolutionary games, one would ascribe 
strategies and payoffs to the ant colony, while ascribing strategies 
and payoffs to the individual hyenas of a pack."
In the case of ants, the proposed solution is thus to model the
ant colony as a player. Within the CEG paradigm, this would
mean that we have to consider interactions between ant colonies.
This however does not allow us anymore to
model behavior at the level of the individual. 

In this work  we present a new model for evolutionary games, in which the concept of the agent as a single individual is substituted by that of the agent as a whole group of individuals.  This new concept, named Group Equilibrium Stable Strategy (GESS),  allow   to model competition between individuals
in a population in which the whole group shares a common utility. Even if we still consider pairwise interactions among individuals, our perspective is completely different: we suppose that individuals are simple actors of the game and that the utility to be maximized is the one of their group.    Our study of evolutionary games under the altruism inside each group, is built around the ESS. We begin by defining the GESS, deriving it in several ways and exploring its major characteristics.  The main focus of this paper is to study how this new concept changes the profile of population  and to 
explore  the relationship between GESS and Nash equilibrium or ESS.  We characterize through the study of many GESS  and we show how the evolution and the equilibrium are influenced by the groups' size as well as  by their immediate payoff.     We also provide some primary results through an example on multiple access games,  in which any local interaction does not lead to same payoff depending on the type of individuals that are competing, and not only the strategy used.  In  such application, we evaluate the impact of altruism behavior on the performance of the system. 

The paper is structured as follows. We first provide in the next section the needed background on evolutionary games. In the section \ref{NNCEG} we then study the new natural concept GESS and the relationship between  GESS  and ESS or Nash equilibrium. The characterization of the  GESS is studied in section \ref{ANGG}.   Section \ref{Examples} provides some numerical illustration through some famous examples in evolutionary games.  In section \ref{MAC} we study  the multiple access control in slotted Aloha under altruism behavior.  The paper closes with a summary in section \ref{conc}

%

\section{Classical Evolutionary Games and ESS}
\label{CEG}
\indent 
We consider an infinite population of players and we assume that each member of the population has the same set of available pure strategies ${\cal{K}}=\{1, 2,..,m\}$.  We suppose that each individual is repeatedly paired off with an other individual randomly selected within the population.  A player may use a mixed strategy $\mathbf{p}\in \Delta({\cal{K}})$ where  $\Delta({\cal{K}})= \{\mathbf{p}  \in \mathbb{R}^m_{+} | \sum_{i\in \cal{K}} p_i =1\}$. Here $\mathbf{p}$ is a probability measure over the set of actions ${\cal{K}}$.  This is the case where an individual has the capacity to produce a variety in  behaviours.  Alternatively,  the mixed strategy $\mathbf{p}$ can be interpreted as the vector of densities of individuals adopting a certain pure strategy, where $p_i$ is the fraction of the population using strategy $i\in{\cal K}$. However the original formulation of evolutionary game theory were not required to make distinction between population-level and individual-level variability for infinite population \cite{MaynSm-Price}.


Let now focus on the case of monomorphic populations  in which each individual uses a mixed strategy.  We define by $J(\mathbf{p},\mathbf{q})$ the expected payoff for a tagged individual if it
uses a mixed action $\mathbf{p}$ when meeting
another individual who adopts the mixed action $\mathbf{q}$. This payoff
is called "fitness" and actions with larger fitness are
expected to propagate faster in a population. If we define a
payoff matrix $A$ and consider $\mathbf{p}$ and $\mathbf{q}$ to be column vectors, 
then $J(\mathbf{p},\mathbf{q})=\mathbf{p}'A\mathbf{q}$ and the payoff function $J$
is indeed linear in $\mathbf{p}$ and $\mathbf{q}$. A mixed action $\mathbf{q}$ is called a
Nash equilibrium if
\begin{equation}
\forall  \mathbf{p}\in \Delta(\mathcal{K}), \quad J(\mathbf{q},\mathbf{q}) \geq J(\mathbf{p},\mathbf{q}) \label{ess4}
\end{equation}
In evolutionary games the most important concept of equilibrium is the ESS, which was introduced by \cite{MaynSm-Price} as a strategy that, if adopted by most members of a population, it is not invadable by mutant strategies in its suitably small neighbourhood. More precisely, we suppose that  the whole population uses a strategy  $\mathbf{q}$ and that a small fraction $\epsilon$ of individuals (\textit{mutants}) adopts another strategy $\mathbf{p}$. Evolutionary forces are expected to select $\mathbf{q}$ against $\mathbf{p}$  if
  \begin{equation}
\label{ess0}
J(\mathbf{q},\epsilon \mathbf{p}+ (1 - \epsilon ) \mathbf{q} ) > J(\mathbf{p},\epsilon
\mathbf{p}+ (1 - \epsilon ) \mathbf{q} )
\end{equation}
The definition of ESS is thus related to a robustness property
against deviations by a whole (possibly small) fraction of the population. 
This is an important difference that distinguishes the equilibrium
in populations as seen by biologists and the standard Nash equilibrium
often used in economic context, in which robustness is defined against
the possible deviation of a single user. Why do we need the
stronger type of robustness? Since we deal with large populations,
it is likely to expect that from time to time, some group of individuals
may deviate. Thus robustness against deviations by a single user
is not sufficient to ensure that deviations will not develop and end
up being used by a growing portion of the  population.  By defining the ESS  through the following equivalent definition {\rm \cite[Proposition 2.1]{weibull} or \cite[Theorem
6.4.1, page 63]{HS98}}, it's possible to establish the relationship between ESS and Nash Equilibrium (NE). Strategy $\mathbf{q}$ is an ESS if it satisfies the two conditions:
\begin{itemize} 
\item Nash equilibrium condition:
\begin{equation}\label{ess6}
J(\mathbf{q},\mathbf{q}) \geq J(\mathbf{p},\mathbf{q})\qquad \forall\mathbf{p}\in\mathcal{K}.
\end{equation}
\item Stability condition: \\
\begin{equation}\label{ess7}
J(\mathbf{p},\mathbf{q}) =J(\mathbf{q},\mathbf{q}) \Rightarrow
J(\mathbf{p},\mathbf{p})<J(\mathbf{q},\mathbf{p}).\quad \forall\mathbf{p}\neq\mathbf{q}
\end{equation}
\end{itemize} 
The first condition  (\ref{ess6}) is the condition for a Nash equilibrium. In fact, if condition (\ref{ess6}) is satisfied, then the fraction
of mutations in the population will tend to decrease (as it has a
lower fitness, meaning a lower growth rate). Thus the action $q$
is then immune to mutations.  If it does not but if still the
condition (\ref{ess7}) holds, then a population using $q$ is
''weakly'' immune against mutants using $p$. Indeed, if the
mutant's population grows, then we shall frequently have
individuals with action $q$ competing with mutants. In such
cases, the condition $ J(\mathbf{p},\mathbf{p}) <J(\mathbf{q},\mathbf{p})$ ensures that the growth
rate of the original population exceeds that of the mutants. Then an ESS is a refinement of the Nash equilibrium.

\section{New natural concept on evolutionary games }
\label{NNCEG}
In this section we present a new concept for evolutionary games, in which the idea of the player as a single individual is substituted by that of a player as a whole group of individuals. The interactions are among individuals but the objective function, which is maximized, is that of the group they belong to.  We assume that the population  is composed of $N$ groups, $G_i$, $i=1,2,..,N$, where  the normalized size of $G_i$ is noted by $\alpha_i$ with $\sum_{j=1}^N \alpha_i=1$.  

For clarity of presentation, we restrict our analysis to pairwise interactions, where each individual can meet a member of its own group or of a different one. Individuals dispose of a finite set of actions: $\mathcal{K} = \{a_1, a_2, ..,a_M\}$.  Let $p_{ik}$ be the probability that an individual in the group $G_i$ choses an action  $a_k\in \mathcal{K}$; we associate to each group $i$  the vector of probabilities ${\bf p_i}=(p_{i1}, p_{i2},..,p_{iM})$ where $\sum_{l=1}^M p_{il}=1$.   By assuming that each individual can interact with another individual with equal probability, then the expected  utility of a group (player)  $i$ is:
\begin{equation}\label{utiliti}
U_i(\mathbf{p}_i, \mathbf{p}_{-i}) = \sum_{j=1}^N \alpha_j J(\mathbf{p}_i, \mathbf{p}_j),
\end{equation} 
\noindent where $\mathbf{p}_{-i}$ is the profile strategy of other groups and $J(\mathbf{p}_i, \mathbf{p}_j)$ is the immediate expected reward of an individual player adopting strategy $\mathbf{p}_i$ against an opponent playing  $\mathbf{p}_j$.   

\subsection{Group Equilibrium Stable Strategy }

The definition of GESS is related to the robustness property against deviations inside each group. 
There are two possible equivalent interpretations of an $\epsilon -$ deviation in this context:

\begin{enumerate}
\item A small deviation in the strategy by all members of a group. If the group $G_i$ plays according to strategy $\mathbf{q}_i$, the $\epsilon-$ deviation, where $\epsilon\in (0,1)$, consists in a shift to the group's strategy ${\bf \bar p}_i=\epsilon \mathbf{p}_i + (1-\epsilon  \mathbf{q}_i)$; 
\item The second is a deviation (possibly) large of a small number of individuals in a group $G_i$, that means that a fraction $\epsilon$ of individuals in $G_i$  plays a different strategy $\mathbf{p}_i $.
\end{enumerate}

After an $\epsilon -$deviation under both interpretations the profile of the whole population becomes $\alpha_i \epsilon \mathbf{p}_i +\alpha_i (1-\epsilon)  \mathbf{q}_i + \sum_{j\not= i} \alpha_j \mathbf{q}_j$. Then the average payoff of group $G_i$ after mutation  is given by:
\small
\begin{equation}\label{mutant}
\begin{split}
 &U_i(\bar{\mathbf{p}}_i, \mathbf{q}_{-i})=\sum_{j=1}^N \alpha_j J(\bar{\mathbf{p}}_i, \mathbf{p}_j)\\
&=U_i(\mathbf{q}_i,\mathbf{q}_{-i})+\epsilon^2\alpha_i \Omega(\mathbf{p}_i,\mathbf{q}_i) +\epsilon \Big(\alpha_i(J(\mathbf{p}_i,\mathbf{q}_i)\Big.\\&\Big. +J(\mathbf{q}_i,\mathbf{p}_i)-2J(\mathbf{q}_i,\mathbf{q}_i))+\sum_{j\not=i} (J(\mathbf{p}_i,\mathbf{q}_{j})-J(\mathbf{q}_i,\mathbf{q}_j)\Big)
\end{split}
\end{equation}
\normalsize

\noindent where \small $ \Omega(\mathbf{p}_i,\mathbf{q}_i) :=J(\mathbf{p}_i,\mathbf{p}_i) - J(\mathbf{p}_i,\mathbf{q}_i)- J(\mathbf{q}_i,\mathbf{p}_i)) +J(\mathbf{q}_i,\mathbf{q}_i)$. \normalsize

\begin{defi}
\label{def2}
A strategy $\mathbf{q}=(\mathbf{q}_1, \mathbf{q}_2,..,\mathbf{q}_N)$ is a GESS if $\forall i\in\{1,\ldots,N\}$,  $\forall \mathbf{p}_i\not=\mathbf{q}_i$, there exists some $\epsilon_{\mathbf{p}_i}\in(0,1)$, which may depend on $\mathbf{p}_i$, such that for all $\epsilon\in(0,\epsilon_{\mathbf{p}_i})$ 
\begin{equation}
\label{ess1} 
 U_i(\bar{\mathbf{p}}_i,{\bf  q}_{-i})  <  U_i(\mathbf{q}_i,\mathbf{q}_{-i}), 
\end{equation}
where $\bar{\mathbf{p}}_i=\epsilon \mathbf{p}_i + (1-\epsilon) \mathbf{q}_i$.
\end{defi}

Hence from equation (\ref{ess1}), strategy $\mathbf{q}$ is a GESS if the two following conditions hold:
\begin{itemize}
\item $\forall \mathbf{p}_i\in [0,1]^M$ 
\begin{equation}
 F_i(\mathbf{p}_i,\mathbf{q})\defeq \alpha_i \Omega(\mathbf{p}_i,\mathbf{q}_i)-U_i(\mathbf{p}_i,\mathbf{q}_{-i})+U_i(\mathbf{q}_i,\mathbf{q}_{-i})\geq 0,
 \label{cond1-GESS}
 \end{equation} 
\item $\exists \mathbf{p}_i\not=\mathbf{q}_i$ such that: 
\begin{equation} 
\textrm{If}\quad F_i(\mathbf{p}_i,\mathbf{q})=0\Rightarrow 
\Omega(\mathbf{p}_i,\mathbf{q}_i)<0
 \label{cond2-GESS}
\end{equation}  
\end{itemize}

\begin{remark}
\label{rem-GESS-Nash}
The condition (\ref{cond2-GESS}) can be rewritten as 
$$U_i(\mathbf{q}_i,\mathbf{q}_{-i})>U_i(\mathbf{p}_i,\mathbf{q}_{-i})$$
which is exactly the definition  of the strict Nash equilibrium of the game composed by $N$ groups  in which each group maximises its own utility. 
\end{remark}

\subsection{GESS and standard ESS}
Here we analyse the relationship  between the standard $ESS$ and our new concept  $GESS$.
\begin{proposition}
Consider  games whose immedaite expected reward is symmetric, i.e. $J({\bf p},{\bf q})=J({\bf q},{\bf p})$.  Then any ESS is a GESS.
\end{proposition}

\begin{proof}
Let $\mathbf{q}=(q,..,q)$ be an ESS.  Combining the symmetry of the payoff function and equation (\ref{cond1-GESS}), we get:
\small 
\begin{equation*}
\begin{split}
& F_i(\mathbf{p}_i,\mathbf{q})=-\Big(\alpha_i(J(\mathbf{p}_i,{q})+J(q,\mathbf{p}_i)-2J(q,q)\Big.\\&\Big.+\sum_{j\not=i} \alpha_j(J(\mathbf{p}_i,q)-J(q,q)\Big)\\&=-2\alpha_i(J(\mathbf{p}_i,q)-J(q,q))-\sum_{j\not=i}\alpha_j (J(\mathbf{p}_i,q)-J(q,q))\\&=- (1+\alpha_i)(J(\mathbf{p}_i,q)-J(q,q))\geq 0
\end{split}
\end{equation*}
\normalsize
\noindent where the second equality  follows from the symmetry of the payoff function $J$ and the last inequality follows form the fact that $\mathbf{q}$ is an ESS and satisfies (\ref{ess6}). This implies that $\mathbf{q}$ satisfies the first condition of GESS (\ref{cond1-GESS}).    Now assume that  $F_i(\mathbf{p}_i,q) =0$  for some $\mathbf{p}_i\not=\mathbf{q}$, previous equations imply that $J(\mathbf{p}_i,\mathbf{q})=J(\mathbf{q},\mathbf{q})$.  Thus the  second condition (\ref{cond2-GESS}) becomes \small
$\Omega(\mathbf{p}_i,\mathbf{q}_i)=J(\mathbf{p}_i,\mathbf{q})-J(\mathbf{q},\mathbf{q})<0$ \normalsize
which coincide with the second condition of ESS (\ref{ess7}). This completes the proof.
\end{proof}

\subsection{Nash equilibrium and GESS}
In the classical evolutionary games, the ESS is a refinement of a Nash equilibrium and we can see that all ESSs are Nash equilibria but not all Nash equilibria are ESSs.  In order to characterize this relationship in our context, let us define  the game between groups: There are  $N$ players in which each player has a finite set of pure strategies ${\cal{K}}=\{1, 2,..,m\}$.   We define by 
$U_i(\mathbf{q}_i,\mathbf{q}_{-i})$ the utility of player $i$ when using mixed strategy $\mathbf{q}_i$ against a population of players using $\mathbf{q}_{-i}=(\mathbf{q}_1,\ldots,\mathbf{q}_{i-1},\mathbf{q}_{i+1}, \ldots,\mathbf{q}_N)$.

\begin{defi}
\label{def22}
A strategy $\mathbf{q}=(\mathbf{q}_1, \mathbf{q}_2,..,\mathbf{q}_N)$ is a Nash Equilibrium   if $\forall i\in\{1,\ldots,N\}$  
\begin{equation}
 U_i( \mathbf{q}_i,{\bf  q}_{-i})  \geq U_i(\mathbf{p}_i,\mathbf{q}_{-i})
 \label{NE}
\end{equation}
for every other mixed strategy $ \mathbf{p}_i\not= \mathbf{q}_i$ . If it holds for strict inequality, then $\mathbf{q}$ is a strict Nash equilibrium. 
\end{defi}

From the definition of the strict Nash equilibrium,  it is easy to  show that any strict Nash equilibrium is a GESS defined in equation  (\ref{ess1}). But in our context, we address several questions on the relationship  between the GESS, ESS  and the Nash equilibrium  defined in (\ref{NE}).  For simplicity of presentation, we restrict to the case of two-strategies games. Before studying them, we introduce here some definitions that are needed in the sequel.

\begin{defi}
\begin{itemize}
\item A {\bf \it fully mixed strategy} ${\bf q}$  is a strategy such that all actions of each group have to receive a positive probability, i.e., $0<q_{ij}<1$ $\forall (i, j)\in {\cal I}\times {\cal K}$.
\item A mixer (pure) group $i$  is the group that uses a mixed (pure) strategy $0<q_i<1$ (resp. $q_i\in\{0,1\}$).  
\item An equilibrium with mixed and non mixed strategies is an equilibrium in which there is at least one pure group and a mixer group.
\end{itemize}
\end{defi}

\section{Analysis of $N$-groups games with two strategies}\label{ANGG}
We will discuss here $N$-groups games with two strategies. The two possible pure strategies are  $A$ and $B$ and the pairwise interactions payoff matrix  is given by:

$$
P=\bordermatrix{~ & A & B \cr
                  A & a & b \cr
                  B & c & d \cr}, \;
$$
where $P_{ij}$, $i,j=A,\; B$ is the payoff of the first (row) individual if it plays strategy $i$ against the second (column) individual playing strategy $j$.  We assume that both individuals are the same and hence payoffs of the column player are given by the transposed of $P$.  According to the definition of GESS, ${\mathbf q}$ is a GESS if it satisfies the conditions (\ref{cond1-GESS})-(\ref{cond2-GESS}), which can be rewritten as:

\small
\begin{itemize}
\item $\forall \mathbf{p}_i\in [0,1]$, $i=1, ..,N$:
\begin{equation}
\begin{split}
&F_i(p_i,{\mathbf q})=(q_i-p_i) \Big(\alpha_i (J(q_i,1)-J(q_i,0))+\Big.\\&\Big.\sum_{j= 1}^N\alpha_j(J(1,q_j)-J(0,q_j))\Big)\geq 0
 \label{cond1}
\end{split}
 \end{equation} 
\item If $F(p_i,{\mathbf q})=0$ for  some $p_i\not=q_i$, then:
\begin{equation} 
(p_i-q_i)^2 \Delta<0 \Longrightarrow \Delta <0
 \label{cond2}
\end{equation}
\noindent where $\Delta=a-b-c+d$. 
\end{itemize}
\normalsize  

\subsection{Characterisation of fully mixed GESS}
In this section we  are interested in characterising the  full mixed GESS ${\bf q}$. According to (\ref{cond1}), a full mixed equilibrium ${\bf q}=(q_1,\ldots,q_N)$ is a GESS  if it satisfies the condition (\ref{cond2}) where the equality must holds for all $p\in[0,1]$. This  yields to  the following  equation: $\forall i=1,\ldots N$,
\small 
\begin{equation*}
 \alpha_i (J(q_i,1)-J(q_i,0))+\sum_{j=1}^N\alpha_j(J(1,q_j)-J(0,q_j))=0
\end{equation*}
\normalsize
\noindent which can be rewritten as
\begin{equation*}
 \alpha_i \Delta q_i  +  b-d+ \alpha_i (c-d)+\Delta \sum_{j= 1}^N \alpha_j q_j=0
\end{equation*}
This leads to the following expression of the mixed GESS:

\begin{equation}
q_i^*=\frac{ d-b+ \big((1+N)\alpha_i-1\big)(d-c)}{(N+1) \alpha_i \Delta };\\
\label{mixed}
\end{equation}
\begin{proposition}
If $\Delta<0$ and $0<q_i^*<1$, $i=1, \ldots ,N $, then there exists a unique fully  mixed GESS equilibrium given by (\ref{mixed}).
\end{proposition}
We note that the fully mixed  GESS is a strict Nash equilibrium since the condition (\ref{cond2}) is equivalent to the definition of the strict Nash equilibrium (see remark \ref{rem-GESS-Nash}) under the condition $F(p, q_1,\ldots,q_N)=0$, $\forall p\in[0,1]$.

\subsection{Characterisation of strong GESS}
We call a strong GESS an equilibrium  that satisfies the strict inequality (\ref{cond1}) for all groups.  Similarly to the fully mixed GESS, we present here the condition for the existence of a strong GESS.  Note that  all groups have to use pure strategy in a strong GESS.  Without loss of generality, we assume that a pure strong GESS can be represented by  $n_A$, where $n_A\in \{1,...,N\}$ denotes  that the $n_A$ first  groups use $A$ pure strategy and remaining $N-n_A$ groups chose strategy $B$.  For example $n_A=N$ (resp. $n_A=0$)  means that all groups choose pure strategy $A$ (resp. $B$).  

\begin{proposition}\label{prop-SGESS}
If $a\not= c$ or $b\not=d$, then every N-player game with two strategies has a GESS. There are the following possibilities for the \textbf{strong GESS}:

\begin{itemize}
\item[i.] 
If $a-c>\max_i(\alpha_i)\cdot (b-a)$ then $n_A=N$ is a  strong GESS;
\item[ii.]
If  $b-d<\min_i(\alpha_i)(d-c)$ then  $n_A=0$ is a  strong GESS; 
\item[iii.] Let $H(n_a):=\sum_{j=1}^{n_A} \alpha_j (a-c) + \sum_{j=n_A+1}^{N} \alpha_j (b-d)$. If $\alpha_i (d-c) > H(n_a) > \alpha_i(b-a)$ then $n_A$ is a strong GESS. 
\end{itemize}
\end{proposition}

\begin{proof}  In order to prove that a strategy $n_A=N$ is a GESS, we have to impose the strict inequality, i.e.: $\;\; \forall p_i\not=1\mbox{ for } i\in\{1,..,n_A\} \mbox{ and }    \forall p_i\not=0 \mbox{ for } i\in\{n_A+1,..,N\}$
$$
F_i(p_i,1_{n_A}, 0_{N-n_A})>0
$$
\noindent We show here  conditions of the existence  only for $n_A=N$ since the  others one  straightforward follow from the symmetry of the players in the game.   

Consider now that $(n_A=N$) is a strong GESS. The inequality (\ref{cond1}) becomes: $\forall p_i\not=1,\;\;\forall i$
\small
$$
(p_i-1) \Big(\alpha_i (a-b) +\sum_{j=1}^N\alpha_j(a-c)\Big)=(p_i-1) \Big(\alpha_i (a-b) +a-c)\Big)<0,
$$
\normalsize

Since $p_i<1$, one has  $\alpha_i (b-a) <a-c$ $\forall i$.
This completes the proof of (i). To show conditions of the other strong GESSs, we follow the lines of the proof of (i). \end{proof}
%
%
%
%
%
%
%
%
%
%
%
%

\subsection{Characterisation of weak GESS}
We call a weak GESS an equilibrium in which at least one group uses a strategy that satisfies the condition (\ref{cond2}) with equality.  Here we distinguish two types of equilibrium: the equilibrium with no mixer group and the equilibrium with mixer and no mixer groups. Conditions for the equilibrium with no mixed strategy are given by Proposition \ref{prop-SGESS} with at the least one group satisfing it with equality and $\Delta<0$.  In this section we focus only on the equilibrium with mixer and no mixer groups.  
 Without loss of generality, we assume that an equilibrium with mixed and non mixed strategies, can be represented by $(n_A,n_B,\mathbf{q})$ where $n_A$ denotes that group $i$ for $i=1..,n_A$ (resp. $i=n_A+1, ..,n_A+n_B$) uses strategy $A$ (resp. $B$). The remaining groups $N-n_A-n_B$  are mixers in which $q_i$ is the probability to choose the strategy $A$ by group $i$. 
\begin{proposition}\label{prop-MPGESS}
Let either $a\not= c$ or $b\not=d$ and $\Delta<0$. $(n_A, n_B,  \mathbf{q})$ is a weak GESS if:

\small 
\begin{equation}
\left\{ \begin{array}{l l}
&\alpha_i \Delta + d-b+\alpha_i(c-d)+\Delta (\alpha_{n_A}+y) \geq0,\;i=1,..,n_A\\
& d-b+\alpha_i(c-d)+\Delta (\alpha_{n_A}+y) \leq 0,\; i=n_A+1,..,n_B\\
&q_i= \frac{d-b+\alpha_i(d-c)-y\Delta}{\Delta \alpha_i},\; i=n_A+n_B+1,..,N
\end{array}\right.
\label{mixer-non-mixer1}
\end{equation}
\normalsize
where $y=\frac{(N-n_A-n_B)(d-b-\sum_{j=1}^{n_A}\alpha_j )+(d-c) \sum_{j=n_A+n_B+1}^{N}\alpha_j}{\Delta (N-N_A-n_B+1)}$. 
\end{proposition}

\begin{proof} 
We assume that  $(n_A,n_B,\mathbf{q})$  is GESS.  From the condition (\ref{cond1})  we get

\small
\begin{equation}
\left\{
\begin{array}{ll}
 \alpha_i \Delta +  b-d+ \alpha_i (c-d)+\Delta y>0,&i=1,..,n_A\\
 b-d+ \alpha_i (c-d)+\Delta y<0,&i=n_A+1,..,n_B\\
   \alpha_i \Delta q_i  +  b-d+ \alpha_i (c-d)+\Delta y=0,& i=n_A+n_B+1,..,N
 \end{array}
  \right.
  \label{simple}
\end{equation}\normalsize
\noindent where $y=\sum_{i=1}^N \alpha_j q_j$. To compute $q_i$, $i=n_A+n_B+1,..,N$, we add the $N-n_A-n_B$  last equations' left hand sides together in (\ref{simple}), which gives:
\tiny
\begin{equation}
\Delta y -\Delta \alpha_{n_A}+ (N-n_a-n_B)(b-d) + (c-d)\cdot \sum_{j=n_A+n_B+1}^N \alpha_j  +\Delta  (N-n_a-n_B) y=0
\end{equation}  \normalsize
Thus:
 \tiny
\begin{equation}
y=\frac{(N-n_A-n_B)(d-b)+ \Delta \alpha_{n_A}+(d-c) \sum_{j=n_A+n_B+1}^{N}\alpha_j}{\Delta (N-N_A-n_B+1)}
\end{equation}  
\normalsize
This agrees with (\ref{mixer-non-mixer1}), completing the proof of the proposition.
 \end{proof}


%
%
%

\section{Some Examples}
\label{Examples}
In this section we analyze a number of examples with two players and two strategies.

\subsection{Hawk and Dove Game}
One of the most studied examples in EG theory is the Hawk-Dove game, first introduced by Maynard Smith and Price in "The Logic of Animal Conflict".  In this game two  animals compete for a resource of a fixed value.  Each animal follows one of two strategies, Hawk or Dove, where  Hawk corresponds to an aggressive behavior,  Dove  to a non-aggressive one. If two Hawks meet, they fight and one of them gets the resource while the other is injured, with equal probability. A Hawk always wins against a Dove, whereas if two Doves meet they equally share the resource. 
The payoff matrix of the game is the following:
$$
\bordermatrix{~ & H & D \cr
                  H & \frac{1}{2}(V-C)& V  \cr
                  D & 0 & V/2 \cr}                 
$$

\noindent where $C$ represent the cost of the fight, and $V$ is the benefit the palyer get from the resource. We suppose that $C>V$. 

 In standard GT, this example belongs to the \textit{anti-coordination class}, whose games always have two strict, pure strategy NEs and one non-strict, mixed strategy NE. In this case the two pure equilibria are $(H,D)$ and $(D,H)$, and the mixed-one is given by: $q^*=\frac{V}{C}$. The latter is the only ESS: even if the two pure NE are strict, being asymmetric they can't be ESSs.
We set $V=2$ and $C=3$ and  study this game in groups framework, considering two groups of size $\alpha$ and $1-\alpha$.

We obtain that the GESSs and the strict NE always coincides. In particular we observe that:
\begin{itemize}
\item for $0<\alpha<0.25$ the game has one strong GESS $(H,D)$ and a weak GESS $(H,q_2)$ ;  
\item  for $0.25<\alpha<0.37$: one weak GESS $(H,q_2)$;  
\item  for $0.37<\alpha<0.5$ one weak GESS,  $(q^*_1,q^*_2)$;
\end{itemize}

The size of groups has a strong impact on the beahvior of players: in the first interval of $\alpha-$values we remark that the GESS is not unique; when increasing the size of the first group, and thus decreasing that of the second one, the probability that the second group plays aggressively against the pure aggressive strategy of the first one increases until we get into the third interval, where both players are mixers. Here we can clearly observe that the equilibrium $q_1^*$ is decreasing in $\alpha$: as we supposed that an individual can interact with members of its own group, when increasing $\alpha$, the probability of meeting an individual in the same group increases and it leads to a less aggressive behavior.

\begin{figure}[!htb]
    \centering
        \includegraphics[height=5.2cm,width=7cm]{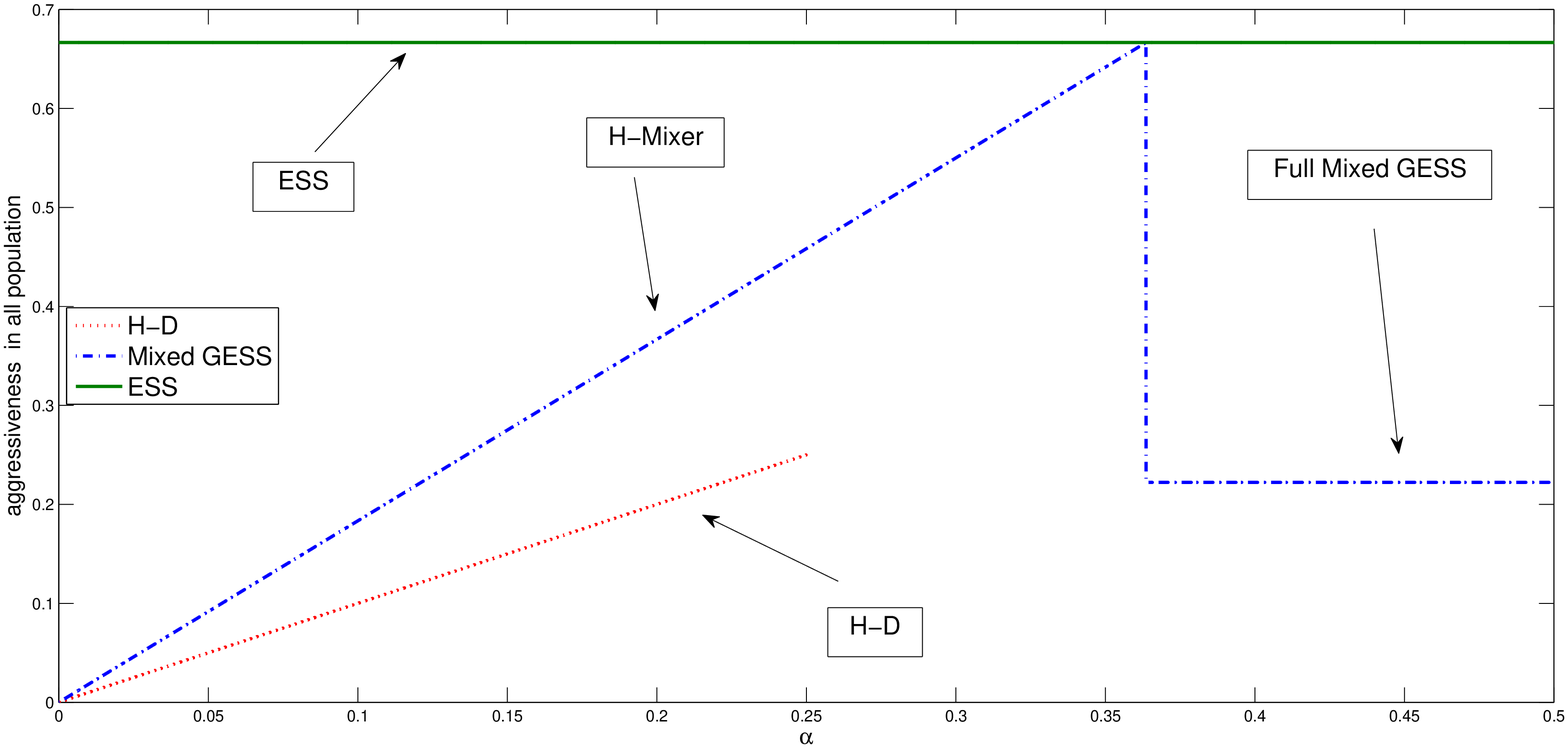}\label{HDalpha.eps}

       \caption{The global level of aggressiveness in the two-groups population for the different GESSs, as a function of $\alpha$ }
\end{figure}

\subsection{Stag Hunt Game}
We now consider a well-known example which belongs to the \textit{coordination class}, the Stag Hunt game. The story behind has been described by J-J. Rousseau: two individuals go out on a hunt; if they cooperate they can hunt a stag; otherwise, hunting alone, a hunter can only get e hare, so collaboration is rewarding for players. It represents a conflict between social and safely cooperation.
The payoff matrix is the following:

$$
\bordermatrix{~ & S & H \cr
                 S & a & b \cr
                  H & c & d \cr}                 
$$

\noindent where $S$ and $H$ stand respectively for Stag and Hare and  $a>c\geq d>b$. Coordination games have two strict, pure strategy NEs and one non-strict, mixed strategy NE, respectively the risk dominant equilibrium $(H,H)$, the payoff dominant one $(S,S)$ and the mixed symmetric one with $q_1^*=q_2^*=\frac{d-b}{a-b-c+d}$.

We set $a=2$, $b=0$, $c=1$, $d=1$ and we look for the equilibria of the two groups gama as a function of $\alpha$. We find that the strict GESSs and the strict NEs don't coincide, as we found strategies, which are strict GESSs but not NEs. The two-groups Stag-Hunt Game only have the pure-pure strict NE $(S,S)$, for all values of $\alpha$, whereas, for the GESSs we find that:  
\begin{itemize}
 \item  for $0<\alpha<0.5$ the game admits two pure-pure strong GESSs: $(S,S)$ and $(H,H)$;
\item for $0.25<\alpha<0.5$ the game admits three pure-pure strong GESSs: $(S,S)$ and $(H,H)$ and $(S,H)$;  
\item  the game doesn't admit any strict mixed NE.  
\end{itemize}


\begin{figure}
    \centering
        \includegraphics[width=7cm]{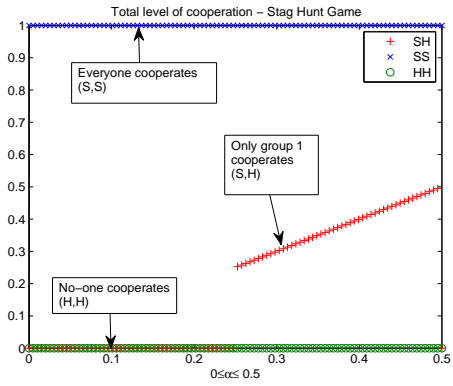}
        \label{CoopSH}

       \caption{The global level of cooperation in the two-groups population for the different GESSs, as a function of $\alpha$ in \ref{CoopSH} and of $x$ in \ref{CoopSHx} }
\end{figure}

\subsection{Prisoner's Dilemma}

We consider another classical example in game theory, the Prisoner's Dilemma, which belongs to a third kind of games, the \textit{pure dominance class}. 

Two criminals are arrested and separately interrogated: they can either accuse the other, either remain silent. If both of them accuse the other (defect) , they will be both imprisoned for 2 years. If only one accuse the other, the accused is imprisoned for 3 years while the other is free. If both remain silent (cooperate), each of them will serve one year in jail. The payoff matrix is the following:
$$
\bordermatrix{~ &C & D \cr
                  C & a & b \cr
                 D & c & d \cr}                 
$$

\noindent where $C$ and $D$ stand respectevely for collaboration and defection and where $c>a>d>b$.

In standard GT, pure dominance class games admit a unique pure, strict and symmetric NE, which also is the unique ESS; in the Prisoner's Dilemma it's $(D,D)$.

We set $a=2$, $b=0$, $c=3$, $d=1$ and we study the two-groups corresponding game. As in the previous example, we find strict GESSs which are not strict NEs. In particular, we have that:
\begin{itemize}
\item $(C,C)$ is always a GESS and a strict NE for all values of $\alpha$;  
\item $(D,D)$ is a GESS  for all values of $\alpha$ but it is never a strict NE;
\item  $(C,D)$ (symmetrically $(D,C)$ ) is always a GESS and a strict NE for $0.5<\alpha<1$ (symmetrically $0<\alpha<0.5$);
\item the game doesn't admit mixed GESSs;
 \end{itemize}

The two groups game thus admits three pure GESSs and  two pure strict NEs for all values of $\alpha$.

\begin{figure}[!htb]
\centering
\includegraphics[width=6cm]{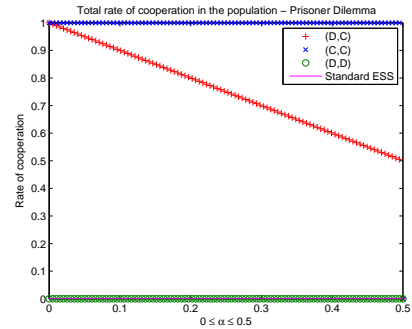}
\caption{Total rate of collaboration for different GESSs as a function of group 1 size $\alpha$ in the prisoner's Dilemma} 
\label{pd}
\end{figure}

%
%
%
%

\section{Multiple Access Control}
\label{MAC}

In this section we briefly introduce a refinement of our model, which will be further developed in future works. 
We modify the group utility function defined in (\ref{utiliti}) by supposing that the immediate payoff matrix differs if the interacting individuals belong to the same group or to two different ones. 

The utility function of a group $i$ playing $q_i$ against a population profile $q_{-i}$ can be written as follows:
\begin{equation}
U(\mathbf{q_i},\mathbf{q_{-i}})=\alpha_iK(\mathbf{q_i},\mathbf{q_i})+\sum_{j\neq i}\alpha_iJ(\mathbf{q_i},\mathbf{q_{j}}),
\end{equation}

\noindent where $K(\mathbf{p},\mathbf{q})$ indicates the immediate expected payoff of an individual playing p against a member of its own group using q , whereas $J(\mathbf{p},\mathbf{q})$ is the immediate expected payoff associated to interactions among individuals of different groups.

We present a particular application of this model in Aloha system in which a large population of mobiles interfere with each other through many local pairwise interactions.  We assume that the population is decomposed into $N$ groups  $G_i$, $i=1,2,..,N$ of normalized size $\alpha_i$ with $\sum_{j=1}^N \alpha_i=1$.  Mobiles are randomly placed over a plane and  matched through pairwise interactions where each mobile decides either to transmit $(T)$ or to not transmit $(S)$ a packet to a
receiver when they are within transmission range of each other. 
Interferences occur as in the Aloha protocol: if more than one neighbor of a receiver transmits a packet at the same time then there is a collision and the transmission fails. The channel is ideal for transmission and all errors are due to collisions.
Let $\mu$ be the probability that a mobile $k$ has its receiver $R(k)$ within its
range. When a mobile $k$ transmits, all mobiles within a circle of radius $R$ centered at
node $R(k)$ cause interference to $k$ for its transmission to $R(k)$, so that more than
one transmission within a distance $R$ of the receiver in the same slot cause a collision and the
loss of mobile's $i$ packet at $R(k)$.

A mobile  of group $i$ may use a mixed strategy $\mathbf{p}_i=(p_i,1-p_i)$   where  $p_i$ is the  probability  to choose the action $(T)$.  
Let $\gamma$ denotes the probability that a mobile is alone in a given local interaction and the tagged mobile does not know whether there is another transmitting mobile within its range of transmission.

Let $P_1$ (resp. $P_2$) be the immediate payoff matrix of interactions among individuals belonging to the same group (resp. of two different ones):

\small
$$ P_1\equiv\bordermatrix{~ & T & S \cr
                  T & -2\delta & 1-\delta \cr
                  S & 1-\delta & 0 \cr}, \;\qquad P_2\equiv\bordermatrix{~ & T & S \cr
                  T & \-\delta & 1-\delta \cr
                  S & 0 & 0 \cr}. \;
$$
\normalsize
where $0<\delta<1$ is the  cost of transmission.  The definition of $P_1$ implies that when two mobiles of the same group $i$ interact, any successful transmission is equally rewarding for the group $i$.   
The resulting expected payoff functions of a mobile playing $\mathbf{q}_i$ against a member belonging to its own group and to a different one, using respectively the same strategy $\mathbf{q}_i$ and a different one $\mathbf{q}_j$ are the following: 
\small

\begin{equation*}
\begin{split}
K(\mathbf{q}_i,\mathbf{q}_i)&=\mu\left[ q_i[\gamma(1-\delta)+(1-\gamma)((1-\delta)(1-q_i)-2\delta q_i)]\right.\\&\left.+(1-\gamma)(1-\delta)(1-q_i)q_i\right]\\&=\mu q_i[(1-\delta)(2-\gamma)-2(1-\gamma)q_i]
\end{split}
\end{equation*}

\begin{equation*}
\begin{split}
J(\mathbf{q}_i,\mathbf{q}_j)&=\mu q_i[\gamma(1-\delta)+(1-\gamma)((1-\delta)(1-q_j)-\delta q_j]\\&=\mu q_i[1-\delta-(1-\gamma)q_j]
\end{split}
\end{equation*}

\normalsize
The expected throughput  of group $i$ is then given by: 

\small
\begin{equation}
\begin{split}
U(\mathbf{q}_i,\mathbf{q}_{-i})&=\mu q_i[1-\delta+(1-\gamma)(\alpha_i(1-\delta-q_i)-\sum_{j=1}^N\alpha_j q_j)]
\end{split}
\end{equation}
\normalsize

By following the same analysis as in section \ref{NNCEG}, the strategy  $\mathbf{q}$ is a GESS  if $\forall i=1,\ldots N$ the two following conditions are satisfied:

%
%
%
%

\begin{enumerate}
\item $F'_i(\mathbf{p}_i,\mathbf{q})\equiv(q_i-p_i)[1-\delta+(1-\gamma)(\alpha_i(1-\delta-2q_i)-\sum_{j=1}^N\alpha_jq_j)]\geq 0\quad\forall p_i,$
\item If $F'_i(\mathbf{p}_i,\mathbf{q})=0$ for some $p_i\neq q_i$, then  $(p_i-q_i)^2(1-\gamma)\alpha_i >0\quad  \forall p_i\neq q_i.$
\end{enumerate}

We observe that the inequality $(p_i-q_i)^2(1-\gamma)\alpha_i >0$ holds for all values of the parameters which implies that the second condition is always satisfied and that the first condition is sufficient to guarantee the existence of a GESS. In the following proposition we give a characterization of the GESSs of the presented MAC game.  Without loss of generality, we reorder the groups so that $\alpha_1\leq\alpha_2\ldots\leq\alpha_N$.

\begin{proposition} We find that:

\begin{itemize}
\item The pure symmetric strategy $(S,\ldots,S)$ is never a GESS.
\item If a fixed group $G_i$ adopts pure strategy $T$, then at the equilibrium, all smaller groups transmit. If the bigger group $G_N$ use strategy $T$ at the equilibrium, then $\gamma\geq\bar{\gamma}$. 
\item If a fixed group $G_i$ adopts pure strategy $S$, then at the equilibrium, all smaller groups also use $S$.
\item If a fixed group $G_i$ adopts an equilibrium mixed strategy $q_i\in]0,1[$, then if $q_i>\frac{1-\delta}{2}$, at the equilibrium all smaller groups use pure strategy $T$, whereas if $q_i<\frac{1-\delta}{2}$, smaller groups play $S$.
\item The game admits a unique fully mixed GESS $\mathbf{q^*}=(q_1^*,\ldots,q^*_N)$, given by: 
\small
\begin{equation}
\begin{split}
q^*_i&=\frac{(1-\delta)(1+\gamma+(1-\gamma)(2+N)\alpha_i)}{2(N+2)(1-\gamma)\alpha_i}\end{split}
\label{MACGESS}
\end{equation}
\normalsize
under the condition: $\gamma<\underline{\gamma}$.
\end{itemize}

\noindent The thresholds $\underline{\gamma}$ and $\bar{\gamma}$ are defined as follows: 
\small
$$\underline{\gamma}\equiv\min_{\alpha_i}\frac{\alpha_i(N+2)(1+\delta)-(1-\delta)}{\alpha_i(N+2)(1+\delta)+(1+\delta)},$$
$$\bar{\gamma}\equiv\max_{\alpha_i}\left(1-\frac{1-\delta}{\alpha_i(\delta+1)+1}\right).$$

\normalsize

\end{proposition}

\begin{proof} A strategy $\mathbf{q}=(q_1,\ldots, q_N)$ is a GESS if $\forall i=1,\ldots,N$, the condition  $F'_i(\mathbf{p}_i,\mathbf{q})\geq 0$ is verified  $\forall p_i\in[0,1]$.
\begin{itemize}
\item If $q_i=0\;\forall i\Rightarrow F_i(\mathbf{p}_i,\mathbf{0})=-p_i[1-\delta+(1-\gamma)(1-\delta)\alpha_i<0\;\forall p_i$, which proves that $(S,\ldots,S)$ is never a GESS.
\item Let $\mathbf{q}$ be a GESS such that $q_i=1$ for a fixed $i$. This implies that $1-\delta-(1-\gamma)(\alpha_i(1+\delta)+Y]\geq0$, with $Y\:=\sum_{j=1}^N\alpha_jq_j)$. Then, if $\alpha_j<\alpha_i$ we have that  $1-\delta-(1-\gamma)(\alpha_j(1+\delta)+Y]\geq 1-\delta-(1-\gamma)(\alpha_i(1+\delta)+Y]\geq 0 $ and thus $q_j=1$ satisfy the GESS condition. If all the groups transmit, then the condition for the GESS is satisfied iff $\forall i :$ $1-\delta-(1-\gamma)((1+\delta)\alpha_i-1)\geq 0$, and thus $\gamma\geq 1-\frac{1-\delta}{\alpha_i(\delta+1)+1}$.
\item Let $\mathbf{q}$ be a GESS such that $q_i=0$ for a fixed $i$. This implies that $1-\delta+(1-\gamma)(\alpha_i(1-\delta)-Y]\leq0$. If $\alpha_j<\alpha_i$, $1-\delta+(1-\gamma)(\alpha_i
j(1-\delta)-Y]\leq1-\delta+(1-\gamma)(\alpha_i(1-\delta)-Y]\leq0$ and thus $F'_j(\mathbf{p}_i,\mathbf{q})\geq0$.
  \item Let $\mathbf{q}$ be a GESS such that $q_i\in]0,1[$ for a fixed $i$.Then, if for a $j<i$, $q_j=0$ (resp. 1), $F'_j(\mathbf{p}_i,\mathbf{q})\geq 0$ implies that $q_i>\frac{1-\delta}{2}$ (resp. $q_i>\frac{1-\delta}{2}$). 
\item  Let $\mathbf{q}$ be a fully mixed GESS. Then, $\forall i$: $1-\delta+(1-\gamma)(\alpha_i(1-\delta-2q_i)-Y)=0$. 
After some albebra we thus obtain that $Y=\frac{(1-\delta)(N+1-\gamma)}{(1-\gamma)(N+2)}$, and by substituting it in the previos equations we obtain the expressions of $q_i$. By imposing that $0<q_i<1\;\forall i$ we obtain the condition $\gamma<\underline{\gamma}$.
\end{itemize}
\end{proof}

%
%
%
%
%
%

As an example, we consider a two groups MAC game, in which we fix a low value of the cost of transmission, $\delta=0.2$, and groups' sizes $\alpha_1=\alpha=0.4$, $\alpha_2=1-\alpha=0.6$, and we vary the value of the parameter $\gamma$. The game admits three different equilibria, depending on  $\gamma$: a fully mixed GESS $\mathbf{q}^*_M$, a pure GESS  $\mathbf{q}^*_P$, and a pure-mixed one $\mathbf{q}^*_{PM}$. In figure \ref{figMAC1} we consider the fully mixed and the pure GESS. We have that for $\gamma<\underline{\gamma}=0.3$ the game admits a GESS $\mathbf{q}^*\mathbf{q}_M^*=(q_1^*,q_2^*)$, whose components are plotted. Then, for $\gamma=\bar{\gamma}>0.53$, $\mathbf{q}^*\mathbf{q}_P^*=(T,T)$. We also plot the value of $q^*_{std}:=\min(1,\frac{1}{1-\gamma}-\Delta$. We note that fully mixed equilibrium strategies adopted by the two groups, $q_1^*, q_2^*$, are both lower then $q^*_{std}$.

\begin{figure}[!htb]
    \centering
        \includegraphics[height=5.2cm,width=7cm]{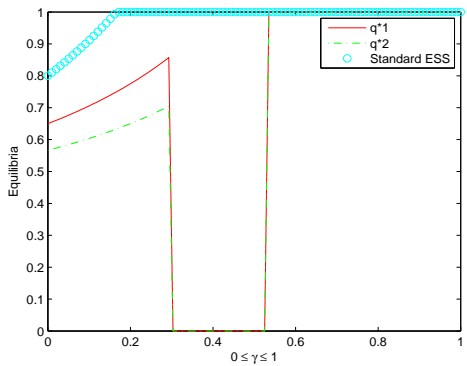}
       \caption{The value of the equilibrium strategy $q^*_1$ and $q_2^*$ in a two groups MAC game as a function of $\gamma$ for $\alpha=0.4$ compared to $q^*_{std}$ .}\label{figMAC1}
\end{figure}

In figure \ref{figMAC2} we plotted the value of the second component of the pure-mixed GESS of the game: $(T,q_T)$, which exists only for $0\leq\gamma<0.4$, and we compare it to $q^*_{std}$.  We note that, for the second group the probability of transmitting is always lower than in the standar game, whereas 

\begin{figure}[!htb]
    \centering
        \includegraphics[height=5.2cm,width=7cm]{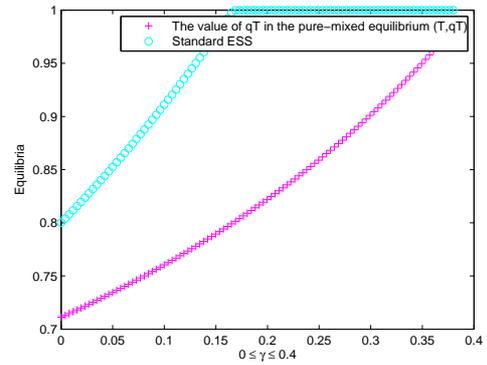}
       \caption{The value of the equilibrium strategy $q^*_2$ of the second group in the pure-mixed equilibrium $(T,qT)$ as a function of $\gamma$ for $\alpha=0.85$ compared to $q^*_{std}$ .}\label{figMAC2}
\end{figure}

Let $p_S(\mathbf{p})$ be the probability of a successful transmission in a population under profile $\mathbf{p}$. For $N=2$, we have: 
\small
\begin{equation*}
\begin{split}
p_S(\mathbf{p})&=\mu[\gamma (\alpha p_1+(1-\alpha)p_2)] +(1-\gamma)(2\alpha^2 p_1(1-p_1)+\\&+\alpha(1-\alpha)((1-p_2)p_1+(1-p_1)p_2)+2(1-\alpha)^2p_2(1-p_2))].
\end{split}
\end{equation*}
\normalsize

In figure \ref{figPsP} we plotted the value of $p^*_S=p_S(\mathbf{q}_M^*)$ as a function of $\gamma$ for $\alpha=0.4$. We note that, $\gamma<\underline{\gamma}$, even if in the groups game the probability to transmit is lower at the equilibrium.

\begin{figure}[!htb]
    \centering
        \includegraphics[height=5.2cm,width=7cm]{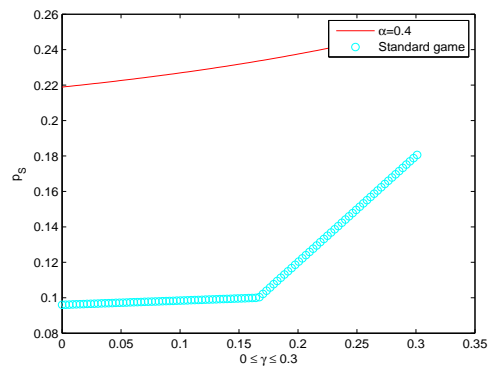}
       \caption{The probability $p_S(\mathbf{q}_M^*)$ of a successful transmission for the fully mixed GESS in a two groups MAC game at the equilibrium as a function of $\gamma$ .}\label{figPsP}
\end{figure}

\section{Conclusions}
\label{conc}
In this work we presented a new concept of Evolutionarily Stable Strategy in a group-players framework, the GESS, exploring its relation with the Nash equilibrium and with the standard ESS.  By analyzing some examples of two players and two strategies games, we observed how the presence of groups impacts the behavior of individuals and changes the structure of the equilibria.
We then introduce a slightly different situation by redefining the utility of a group, in such a way to consider different utilities for interactions among members of the same group or of a different one. Through a MAC example we showed how the presence of groups can favor cooperative behaviors. 
There are still many issues open for future studies. We are now studying the replicator dynamics in the group-players population in order to investigate the relation between our GESS and the rest point of such dynamics. 
At the theoretical level, we want to further deepen the study of the utility of a group, considering selfish and altruistic components.

\end{document}